\documentclass[a4paper,twocolumn]{article}
\usepackage[pdftex]{graphicx}
\usepackage{amsmath}
\usepackage{amssymb}
\usepackage{amsthm}
\usepackage[utf8]{inputenc}
\usepackage[T1]{fontenc}
\usepackage{hyperref}
\usepackage{smartref}
\usepackage{todonotes}
\usepackage{fullpage}
\usepackage{epstopdf}
\usepackage{subcaption}
\usepackage[inline]{enumitem}

\newtheorem{prop}{Proposition}

\newtheorem{theorem}{Theorem}
\newtheorem{definition}{Definition}

\newtheorem{remark}{Remark}
\newtheorem{corollary}{Corollary}
\usepackage[toc,page]{appendix}

\usepackage{color}

\renewcommand\vec{\boldsymbol}

\newcommand{\norm}[1]{\left\lVert#1\right\rVert}
\newcommand{\scal}[2]{\langle #1,#2\rangle}

\newcount\colveccount
\newcommand*\colvec[1]{
        \global\colveccount#1
        \begin{pmatrix}
        \colvecnext
}
\def\colvecnext#1{
        #1
        \global\advance\colveccount-1
        \ifnum\colveccount>0
                \\
                \expandafter\colvecnext
        \else
                \end{pmatrix}
        \fi
}

\begin{document}

\title{On transformations and graphic methods of algebraically 3 dimensional force, velocity and displacement systems}
\author{Tam\'as Baranyai}
\maketitle

\abstract

In engineering practice one often encounters planar problems, where the corresponding vector space of forces, velocities or (infinitesimal) displacements is three dimensional. This paper shows how these spaces can be factorized, such that the arising equivalence classes correspond to points and lines of action of the forces / velocities / displacements in the (projective) plane. It is shown how the study of projective transformations and dualities of planar mechanical systems is closely related to the study of linear maps of these spaces. A few past results are analysed and sometimes extended to show the power of this description.

\section{Introduction}
Given an object in three dimensional space, both all the possible forces it can be subjected to, and all the possible velocities it can have are describable with 6 dimensional vector-spaces. If one considers a linear approximation of the movements of the object the displacement systems are also describable with a 6 dimensional vector-space. In the past century it was shown \cite{ball1900treatise,klein1939elementary} how these 6 numbers can be considered projective homogeneous line coordinates, such that they represent the line of action of the force and axis of rotation of the angular velocity or infinitesimal rotation. This description, called screw theory has become well known in robotics \cite{davidson2004robots} or line geometry \cite{pottmann2001computational}, but not so much in civil engineer circles; in spite of rigidity theorists using it to prove the projectie invariance of the rigidity of bar-joint frameworks and plate (panel, sheet) structures \cite{crapo1982statics}. A formulation for planar motions \cite{whiteley1978introductionI} and forces \cite{whiteley1978introductionII}, similar to the one presented here have been also developed, but it does not appear to have taken root. 

A serious advantage of this connection between objects of projective spaces and forces, velocities and displacements is the ease their transformations can be described and understood. The use of such transformations to engineering purposes is also an old idea \cite{rankine1857transformations}, but this area appears to be of current reserch interest \cite{huerta2010designing,fivet2016projective}. These works however say very little about the changes in the magnitudes of forces, meaning they are of limited use when one intends to optimize structures for mechanical performance. The description provided here directly gives how forces of a structure change when the geometry of the structure is subjected to a projective transformation, and is useful for optimization in this sense.\\

In this paper we restrict ourselves to mechanical systems where the forces, velocites and displacements can be described with three dimensional vector-spaces. The objects forming such mechanical systems will always lie in a plane, thus allowing the connection with the projective plane. The relevant forces can be either coplanar ($\mathcal{F}_p$) or orthogonal to the plane ($\mathcal{F}_c$), which is concurrent in the projective sense: the lines of action of the forces meet at infinity. (In certain cases the orthogonality is not a criterion and sometimes it is more convenient to consider different concurrent systems, see \cite{baranyai2018duality}.) The same can be said for relevant velocities ($\mathcal{E}$) and infinitesimal displacements ($\mathcal{D}$). When considering a mechanical system these physical quantities have to complement each other, meaning we have two meaningful possibilities. One is the usual \emph{planar system} $(\mathcal{D}_c,\mathcal{E}_c,\mathcal{F}_p)$: bodies being subjected to coplanar forces and they are moving in the same plane. The  movement is conveniently given by rotation vectors perpendicular to the plane. In case of the other, \emph{complementary planar system} $(\mathcal{D}_p,\mathcal{E}_p,\mathcal{F}_c)$ the bodies are subjected to forces perpendicular to a given plane, while the motions are given by rotation and displacement vectors lying in the plane. This is not merely a theoretical setting, since in civil engineering, one often has to create planar space covers. In fact one of the first works \cite{TARNAI1989duality} investigating planar dualities is about a typical planar space cover, namely grillages (a planar networks of beams).\\

A clear cut linear algebraic formulation is given of the correspondence between forces / velocities / displacements and their points / lines of attack in the projective plane. Using this formulation it is shown how one can factorize the linear transformations of the three dimensional vector-spaces of these physical quantities into equivalence classes that correspond to projective transformations and dualities of the plane. With this, given any linear map acting on the space of said physical quantities we can instantly see its effect on the geometry; and given a geometric transformation we can instantly see the set of linear maps on the physical quantities that could correspond to it. The power of this description is supported by the following additional results: 

A known combinatorical result \cite{whiteley1991weavings} stating the existence of a spherical polyhedron corresponding to a self stress of a grillage, whose projection is the dual grid of the grillage is analytically supplemented. 

The fact that rigidity is a projective invariant is proven for general planar force systems and structures, not just bar and joint frameworks and grillages as it appears to have been so far.

It is shown that the moment functional (stress function) of graphic statics can be dualized to describe the appropriate velocities of certain mechanisms. A consequence of this is that a three dimensional diagram depicting the velocity state of the mechanism can be created using this velocity functional and a three dimensional projective duality, analogously to Maxwell's construction \cite{maxwell1864,maxwell1870}.

A theorem dual to the Aronhold-Kennedy theorem of kinematics is presented, which appears to be stronger than the one present in current literature \cite{shai2006study}, as it does not require concepts additional to the basic concept of the force.

It is shown how under certain conditions an accurate moment diagram can be created using the reference line of the structure and the curve of the resultant forces (thrust line).

\section{Preliminaries}
As the main result of this work stems from pairing geometrical and physical objects, we start with the mathematical definitions. Due to the nature of the mechanical properties, we restrict ourselves to the real projective plane, which we will denote with $PG(2)$.  
\subsection{Elements of the projective plane}
Consider $\mathbb{R}^3$, and the following equivalence relation
\begin{align}
p\sim q \iff q \in\{\lambda p  \ \vert \ \lambda \in \mathbb{R}\setminus \{0\} \}\label{eq:equivalence}
\end{align}
($p,q\in \mathbb{R}^3$).
The set of all such $q$ is called the equivalence class of $p$ (denoted with $p_{\sim}$), while the vectors themselves are called representants (or representatives) of the equivalence class. We can consider each equivalence class as a point in $PG(2)$. How one embeds the Euclidean space into the projective one carries certain freedom in terms of coordinates. In this paper point $\vec{p}\in \mathbb{R}^{2}$ is mapped to the equivalence class $(\vec{p},1)_{\sim} $. Correspondingly, points at infinity in direction $\vec{u}$ will be represented with $(\vec{u},0)_{\sim}$.  

One way of looking at this description is that we have embedded the $x,y$ plane into $\mathbb{R}^3$ with $z=1$. This image is useful to understand what lines are in this setting: planes in $\mathbb{R}^3$, passing through the origin and intersecting the $z=1$ plane in a line. Conveniently, these planes can be represented by any of their normal vectors, which differ in length but not direction. As such, lines of the projective plane can also be thought of as equivalence classes, such that line $l_{\sim}$ contains point $p_\sim$, if and only if $\scal{l}{p}=0$ holds (the usual scalar product). It can be seen, how the choice of representants does not influence this relation.

As points are identified with $1D$ subspaces while lines with $2D$ subspaces of $\mathbb{R}^3$, we have the following: three points are collinear if their representing vectors are linearly dependent. Three lines are concurrent if their representants are linearly dependent. Again, the choice of representants does not effect this relation.

We will consider two types of transformations of $PG(2)$, projective transformations and dualities. They are:

\begin{definition}[Projective transformation]
A projective transformation (collineation) is a $PG(2) \rightarrow PG(2)$ map, mapping points to points and lines to lines, such that incidences are preserved.
\end{definition}

\begin{definition}[Duality]
A duality is a $PG(2) \rightarrow PG(2)$ map, mapping points to lines and lines to points, such that incidences are preserved.
\end{definition}

In order to describe these transformations, let $\mathcal{M}_{3 \times 3}$ denote the set of \emph{invertible} $3 \times 3 $  matrices over $\mathbb{R}$, and let us factorize this set into equivalence classes similarly:
\begin{align}
A\sim B \iff A \in\{\lambda B  \ \vert \ \lambda \in \mathbb{R}\setminus \{0\}\label{eq:equivalence2}
\end{align} 
($A,B\in \mathcal{M}_{3 \times 3}$). Denoting the set of all arising equivalence classes with $\mathcal{M}_{3 \times 3}/\sim$, we can have the following two theorems helping us.

\begin{theorem}[\cite{pottmann2001computational}]\label{thm:projectivity}
There is a bijection between equivalence classes of $\mathcal{M}_{3 \times 3}/\sim$  and the projective transformations of $PG(2)$.
\end{theorem}

\begin{theorem}[\cite{pottmann2001computational}]\label{thm:duality_matek}
There is a bijection between equivalence classes of $\mathcal{M}_{3 \times 3}/\sim$   and the dualities of $PG(2)$.
\end{theorem}

In practice we will have to take into account, whether we are transforming points or lines.  If (in either of the cases) points are transformed as $p_{\sim} \mapsto p_{\sim} P_{\sim}$ then lines are transformed as $l_{\sim} \mapsto l_{\sim} P^{-T}_{\sim}$. It can be seen that how this preserves incidence, as $\scal{p P}{l P^{-T}}=\scal{p PP^{-1}}{l}=\scal{p}{l}$. It is also apparent how the choice of representants does note effect this relation.

\subsection{Planar system}
This system consists of coplanar bodies moving in the plane they lie in and being subjected to coplanar forces. One way of describing the movements is with rotations around axes perpendicular to the plane of interest.

Let us have a coordinate system with axes labelled $x,y,z$ such that the plane of interest is the $x,y$ plane. The velocity state of any body moving in the plane can be described with the triplet $(v_x,v_y,\omega_z)$ where $v_x$ and $v_y$ describe the translational components of the motion and $\omega_z$ the rotational component around the origin. To each motion belongs a single point $r$ in the plane, where the velocity vector is zero: the body rotates around this point. It's coordinates can be calculated as $r_x=-v_y /\omega_z$ and $r_y=v_x/\omega_z$, provided $\omega_z\neq 0$. From this it is apparent that we can consider the triplet 
\begin{align}
e:=(-v_y,v_x,\omega_z)\in \mathcal{E}_c
\end{align}
both a 3 dimensional vector giving the velocity state of a planar body and a representant of the equivalence class $e_{\sim}$ that is corresponding to the point in $PG(2)$ around which the body is rotating. We can also see that rotation around an ideal point (at infinity) means pure translation. The equivalence class consists of velocities that are scalar multiples of each other, corresponding to rotations around the same point, with different speeds.\\

In numerous fields of mechanics the assumption of infinitesimal rotations is common practice. Such displacements can also be identified with projective coordinates, in this displacement system with point coordinates. The description proposed is

\begin{align}
d:=(-d_y,d_x,\phi_z) \in\mathcal{D}_c.
\end{align}
The motion of the body is described with translations $d_x$ and $d_y$ and rotation around the origin with angle $\phi_z$ (one can check how the assumption of infinitesimal rotations, i.e. $\text{sin}(\phi)\approx \phi$ and $\text{cos}(\phi)\approx 1$ makes the order of these transformation irrelevant). Similarly to the angular velocities, $d$ is projective homogeneous coordinate description of the point around which the body ends up rotating.\\

We can describe any force acting in the plane with the triplet $(F_{x},F_{y},M_z)$, where $F_{x}$ and $F_{y}$ are the projections of the force to the coordinate axes and $M_z$ is the moment of the force with respect to an axis perpendicular to the plane passing through the origin (with respect to the origin for short). 

Now let $f \in \mathcal{F}_p$ be defined as
\begin{align}
f:=(-F_{y},F_{x},M_z)\in \mathcal{F}_p
\end{align}  
and let us note, that the the scalar product $\scal{(q_x,q_y,1)}{f_i}=-q_xF_{y}+q_yF_{x}+M_z$ is the moment of the force with respect to the point $(q_x,q_y)$, where  $(q_x,q_y,1)$ can be thought of as carefully chosen representant from the equivalence class $(q_x,q_y,1)_{\sim}$. This shows how the line of action of $f$ is precisely the equivalence class $f_{\sim}$.

\subsection{Complementary planar system}

In this system bodies lying in the $x,y$ plane are subjected to forces pointing in the $z$ direction. We can describe them with their $z$ directional component $F_z$ and their moments $M_x$  and $M_y$ with respect to the coordinate axes. A calculation similar to above shows that if we order them as

\begin{align}
f:=(-M_y,M_x,F_z)\in \mathcal{F}_c
\end{align}
the triplet not only uniquely represents each force, but can be considered homogeneous coordinates for the point in the plane where the force is acting.

Also similarly to what has been discussed above, each instantaneous velocity causing $z$ directional motion of the points lying in the $x,y$ plane can be represented by the triplet

\begin{align}
e:=(-\omega_y,\omega_x,v_z)\in \mathcal{E}_p
\end{align}

where $\omega_x$ and $\omega_y$ are components of the angular velocity vector, and $v_z$ is the velocity if the origin in the $z$ direction. The triplet 

\begin{align}
d:=(-\phi_y,\phi_x,d_z)\in \mathcal{D}_p
\end{align}

is defined similarly, but represents displacements. Both $d \in \mathcal{D}_p$ and $e \in \mathcal{E}_p$ can be considered as projective line coordinates, representing the lines around the body is rotating.

\subsection{Inner products}\label{sec:inner}
As one would expect, the properties of the inner products of $6D$ screw theory hold in this more compact form. they are:
\begin{itemize}
\item $\scal{f_c}{d_p}$ is the work of force $f_c$ on displacement $d_p$
\item $\scal{f_p}{d_c}$ is the work of force $f_p$ on displacement $d_c$
\item $\scal{f_c}{e_p}$ is the power of force $f_c$ given velocity $e_p$
\item $\scal{f_p}{e_c}$ is the power of force $f_p$ given velocity $e_c$
\end{itemize}
If we restrict ourselves to appropriately chosen point and line coordinates, we can conveniently express equilibrium and compatibility equations this way. The appropriately chosen point coordinates $p=(p_1,p_2,p_3)$ have to satisfy $p=(p_1,p_2,1)$ or $p=(p_1,p_2,0)$. The appropriately chosen line coordinates $l=(l_1,l_2,l_3)$ have to satisfy $l=(0,0,1)$ or $l_1^2+l_2^2=1$.  In exchange, we have: 
\begin{itemize}
\item $\scal{p}{f_p}$ is the moment of a planar force with respect to point $p$
\item $\scal{l}{f_c}$ is the moment of an orthogonal force with respect to an axis along line $l$ in the plane (the direction of the axis is given similarly to the way forces are identified with lines).
\item $\scal{l}{e_c}$ is the $l$ directional velocity of any point of a body lying on line $l$, if the body is moving with $e_c$. 
\item $\scal{p}{e_p}$ is the $z$ directional velocity of point $p$ of a body moving with $e_p$. 
\item $\scal{l}{d_c}$ is the $l$ directional displacement of any point of a body lying on line $l$, if the body is displaced with $d_c$. 
\item $\scal{p}{d_p}$ is the $z$ directional displacement of point $p$ of a body displaced with $d_p$.
\end{itemize}

One of these will come up exceptionally often in the following, namely:

\begin{definition}[moment functional]
The function
\begin{align}
m(p):=\scal{p}{f_p}.
\end{align}
\end{definition}

The dual nature of the force (both a vector and a linear functional) explains \cite{baranyai2019analytical} many aspects of the projective duality-based constructions of graphic statics originating from the works of Maxwell.

\section{Transformations}
We will now take a look at transformations of these 3 dimensional vector-spaces and see how they are connected to the geometrical changes of the structures they belong to. This will be done considering $\mathcal{E}$ and $\mathcal{F}$ only, at any point one may substitute $\mathcal{D}$ in place of $\mathcal{E}$ if one wishes to consider infinitesimal displacements. We will consider two types of transformations: linear maps and congruences.

Linear transformations uniquely correspond to elements of the set $\mathcal{M}_{3 \times 3}$ which we already factorized into equivalence classes according to the projective equivalence relation \eqref{eq:equivalence2}, we extend this factorization to transformations of these physical quantities naturally. In other words, denoting  $\mathcal{U} \rightarrow \mathcal{V}$ linear maps (where $\mathcal{U},\mathcal{V}\in \{\mathcal{D}_c,\mathcal{E}_c,\mathcal{F}_p,\mathcal{D}_p,\mathcal{E}_p,\mathcal{F}_c) \}$) with $L^{u,v}$; to each equivalence class $L^{u,v}_{\sim}$ corresponds a single equivalence class $A_ {\sim}\in \mathcal{M}_{3 \times 3}/\sim$, such that for any $L^{u,v}\in L^{u,v}_{\sim}$ there exists a single $A \in A_{\sim}$ satisfying $L^{u,v}:  \mathcal{U} \ni x \mapsto xA \in \mathcal{V}$. This is unique in the other direction as well.

It is traditional in projective geometry to use different weights on homogeneous coordinates, which do not change the point/line they represent but may influence which point/line the sum (as vectors) of a set of elements represents. Such weight choice is called congruence and we will also extend them to mechanical properties. While linear maps acted on elements of a force or displacement systems the same way, congruences act differently on each element.
\begin{definition}[Congruence of force systems]
Congruence $\Psi$ acts on force system $\{f_i\}$ as\begin{align}
\{f_i\}\mapsto \{\psi_i f_i\} \quad \vert \ \psi_i \in \mathbb{R}\setminus \{0\}.
\end{align}
\end{definition}
We can define congruent velocity and displacement systems similarly. Two things are noteworthy: Congruences do not change the line/point of attack of a physical quantity, as they preserve equivalence classes; and congruences and linear transformations of these systems commute. 

For the sake of brevity, from here on we will neglect the equivalence class sign when talking about points, lines and transformations of $PG(2)$; unless there is explicit reason to show the distinction.

\begin{remark}
In spite of the presented description covering only a finite number of concentrated forces, the statements are valid for the case of distributed forces and thus for statics of a continuum. To see this, one only has to replace the finite sums with the appropriate integrals.
\end{remark}

While in certain cases like trusses and grillages the geometry of the structure is related to the geometry of the force or velocity systems, in other ones the shape of the bodies under consideration is not necessary relevant. In any case, if body $\mathcal{B}$ was subjected to forces $f_i$, then after a transformation there has to exist a body $\mathcal{B}'$ subjected to the transformed forces $f'_i$. An analogous statement can be made about bodies possessing the appropriate velocities. An illustration is presented in Figures \ref{fig:i2} and \ref{fig:i1}. Figure \ref{fig:i2} contains a three jointed structure, comprised of two bodies. The dual structure in Figure \ref{fig:i1} is also comprised of two bodies, but connected with hinges along 3 coplanar lines. The incidence constraints in both cases are $\scal{f_a}{A}=0,$ $\scal{f_b}{B}=0,$ $\scal{f_c}{C}=0$. In the primal example they constrain forces to points $A,B$ and $C$, while in the dual example forces are constrained to lines $A,B$ and $C$. Similarly, the equilibrium of 3 forces is a concurrency condition in the primal example and a collinearity condition in the dual example. 

We should also note, that in certain cases the geometry of the force or velocity system does not uniquely (up to scaling) determine the system. For instance one can imagine a simply supported beam loaded with more than one vertical force, or a truss with statical indeterminacy of degree more than one. A previous work \cite{whiteley1978introductionII} used congruences acting on projective points (joints of the truss) and derived the forces from them through Cayley-algebra. The description did not contain all transformations of static equilibrium. We introduced the congruences of force systems to fix this, as they are precisely the operations leaving the geometry of the force system invariant. They do not preserve static equilibrium in general though. Noting that this behaviour is tied to the structure and its loads, we define
\begin{definition}[Equilibrium preserving congruences of a statics problem]
Congruences of the force system of the problem, such that static equilibrium of all (sub)-bodies of the problem is preserved.
\end{definition}
One can have kinematically indeterminate systems as well, calling for 
\begin{definition}[Compatibility preserving congruences of a kinematics problem]
Congruences of the velocity system of the problem, such that compatibility of all (sub)-bodies of the problem is preserved.
\end{definition}
Every problem has one pair of such congruences, corresponding to the uniform scaling of the forces and velocities involved, which we will call trivial congruence. We will not pursue the identification of equilibrium and compatibility preserving congruences, as this is a property of the problems, while this work focuses on transformations.

\subsection{Projective transformations}
Projective geometry in itself gives operations on lines and points. If we want to relate mechanical problems to each other, we have to extend these operations to mechanical quantities.

\begin{definition}[projective transformation of a statics / kinematics problem]
A transformation of a structure into an other one, such that all the points (and lines as such) of the structure are transformed according to the rules of projective transformations. Points and lines of attack of forces / velocities are also transformed according to the same projective transformation. The image structure is in static equilibrium / is compatible if and only if the original structure was. 
\end{definition}
Note, how the definition says nothing about the magnitudes of the mechanical quantities.

The main result among such transformations is the following theorem:
\begin{theorem}\label{thm:projective}
\begin{enumerate}[label=(\roman*)]
\item[]
\item Invertible linear transformations  $\mathcal{F}_p \rightarrow \mathcal{F}_p$ and $ \mathcal{F}_c \rightarrow \mathcal{F}_c$ preserve static equilibrium. Each projective transformation of a statics problem can be described by the composition of a linear map equivalence class and an equilibrium preserving congruence of the force system; and any such composition gives a projective transformation of the statics problem.  
\item Invertible linear transformations  $\mathcal{E}_p \rightarrow \mathcal{E}_p$ and $ \mathcal{E}_c \rightarrow \mathcal{E}_c$ preserve compatibility. Each projective transformation of a kinematics problem can be described by the composition of a linear map equivalence class and a compatibility preserving congruence of the velocity system; and any such composition gives a projective transformation of the kinematics problem.
\end{enumerate}
\end{theorem}

\begin{proof}
\begin{enumerate}[label=(\roman*)]
\item[]
\item
Consider an invertible linear map $A: \mathbb{R}^3 \rightarrow \mathbb{R}^3$, and a force system $f_i \in \mathcal{F}_p$ (where $i\in \mathcal{I}$, some index-set). The force system $f_iA \in \mathcal{F}_p$ is in equilibrium if and only if $f_i$ was in equilibrium, as the two equilibrium equations are connected as: 
\begin{align}
\sum_{\mathcal{I}} (f_i A)=\sum_{\mathcal{I}} (f_i)A=0A=0.
\end{align}
The map  $A: \mathcal{F}_p \rightarrow \mathcal{F}_p$ maps lines of $PG(2)$ (lines of action of $f_i$) into lines of $PG(2)$ according to the rules of projective transformations. Furthermore, we know from Theorem \ref{thm:projectivity} that there are no more $2D$ real projective transformations than those describable by invertible matrices of this size. If the geometry uniquely (up to scaling) determines the forces, this gives a bijection between equivalence classes of $\mathcal{F}_p \rightarrow \mathcal{F}_p$ linear maps and projective transformations. If the problem has a non-trivial equilibrium preserving congruence of the force system, so does the image; for these maps commute. As such we may use the congruence to arrive at any force systems in equilibrium not attainable from the starting one through linear transformations, but sharing the geometry with one that is attainable. Also, composing congruences with projective transformations will not alter the change of the geometry due to the definition of the congruence, implying that all compositions still correspond to only a projective change in the geometry.\\

\item
Consider velocity system $e_i \in \mathcal{E}_c$ (where $i\in \mathcal{I}$, some index-set). Common compatibility equations can be cast in the form: 
\begin{align}
\scal{l}{e_j-e_k}=0 \label{eq:compatibility_c}
\end{align} 
meaning the relative velocity $e_j-e_k$ ($i,j\in \mathcal{I}$) along a line $l$ has to be zero. One such equation can represent a bar or a sliding joint, two such (independent) equations a fixed joint (lying in the intersection of the two lines) or a slider (one of the lines is the line at infinity), while three equations a fix support. Transforming the velocity system as
\begin{align}
\mathcal{E}_c \ni e \mapsto e A \in \mathcal{E}_c 
\end{align}
also induces a projective transformation of $PG(2)$, this time $A$ describing the change of points. 

Remembering that the transformation can be described on line coordinates with $A^{-T}$, the transformed compatibility equations can be written as
\begin{align}
\scal{lA^{-T}}{(e_j-e_k)A}=0.\label{eq:compatibility_p}
\end{align}  
This implies that the transformed velocity system is compatible if and only if the original velocity system is compatible. This process can be repeated similarly in case of  $\mathcal{E}_p \rightarrow \mathcal{E}_p$ linear maps.

If the structure has no non-trivial compatibility preserving congruence of the velocity system, the bijection between equivalence classes of linear maps of velocities and projective transformations of $PG(2)$ are again given by Theorem \ref{thm:projectivity}. If the structure has a non-trivial compatibility preserving congruence of the velocity system, the same can be said about it as in the statical case above.
\end{enumerate}
\end{proof}

\begin{remark}
The compatibility equations $\scal{p}{e_j-e_k}$ $(e_j,e_k \in \mathcal{E}_p$) mean that the relative velocity $e_j-e_k$ of two bodies is such, that at point $p$ of the plane the overlapping points of the bodies move together in the $z$ direction. Additional equations force bodies move together along a line, or the entire plane. 
\end{remark}
\begin{remark}
In the proof above, the strict choice of $A^{-T}$ is not needed, as another matrix of the equivalence class, say $\lambda A^{-T}$ will satisfy the exact same compatibility equations. 
\end{remark}

The spatial analogues of these results can be stated, in the framework of screw theory. The $4\times4$ transformation matrices of $PG(3)$ determine a subset of $6 \times 6$ matrices acting on line coordinates. While past form finding approaches only used the fact that projective transformations preserve equilibrium, with this description one can immediately see the changes in the magnitudes of forces as well.
\subsection{Dualities}
We can similarly extend the geometric concept of dualities to mechanical quantities. An important distinction between projective transformations and dualities is that the image of the structure is not necessary straightforward, as points are not mapped into points.
\begin{definition}[projective dualities of a statics / kinematics problem]
A transformation of mechanical quantities such that points and lines of attack of forces / velocities are transformed  into lines and points of attack, according to the rules of projective dualities. The image structure is in static equilibrium / is compatible if and only if the original structure was.  
\end{definition}

In this subsection we consider only dualities mapping forces to forces and velocities to velocities. Operations mapping forces to velocities and vice versa are described later. Our main result here is the following Theorem:

\begin{theorem}\label{thm:duality}
\begin{enumerate}[label=(\roman*)]
\item[]
\item Invertible linear transformations  $\mathcal{F}_p \leftrightarrow \mathcal{F}_c$ preserve static equilibrium.
Each projective duality of a statics problem can be described by the composition of a linear map equivalence class and an equilibrium preserving congruence of the force system; and any such composition gives a projective duality of the statics problem.
\item Invertible linear transformations  $\mathcal{E}_p \leftrightarrow \mathcal{E}_c$ preserve compatibility. Each projective duality of a kinematics problem can be described by the composition of a linear map equivalence class and a compatibility preserving congruence of the velocity system; and any such composition gives a projective duality of the kinematics problem.
\end{enumerate}
\end{theorem}

\begin{proof}
\begin{enumerate}[label=(\roman*)]
\item[]
\item
Consider force system $f_i \in \mathcal{F}_p$ and force system $f_iA \in \mathcal{F}_c$ where $A$ is invertible. Just like in case of Theorem \ref{thm:projective}, the linear map preserves statical equilibrium, the only difference to it is that lines of action of forces are mapped to points of action of forces, meaning the effect of the transformation on the geometry of the structure is a projective duality. In case of no non-trivial congruences the bijection between the equivalence classes of linear maps of forces and dualities is given by Theorem \ref{thm:duality_matek}. The effect of congruences is again similar.

\item
Similarly, consider velocity systems $e_i \in \mathcal{E}_c$ and force system $e_iA \in \mathcal{E}_p$. The geometrical transformation induced by $A$, mapping points $e$ to lines $eA$ is a duality, mapping lines $l$ to points $lA^{-T}$. In case of no non-trivial compatibility preserving congruences the bijection between  equivalence classes of linear maps of velocities and projective transformations of $PG(2)$ are again given by Theorem \ref{thm:duality_matek}. The compatibility conditions of form \eqref{eq:compatibility_c} turn into equations of form \eqref{eq:compatibility_p}, meaning compatibility is preserved.
The effect of non-trivial compatibility preserving congruences is similar to the cases discussed above.
\end{enumerate}
\end{proof}

\begin{figure}[h!]
\includegraphics[width=1\columnwidth]{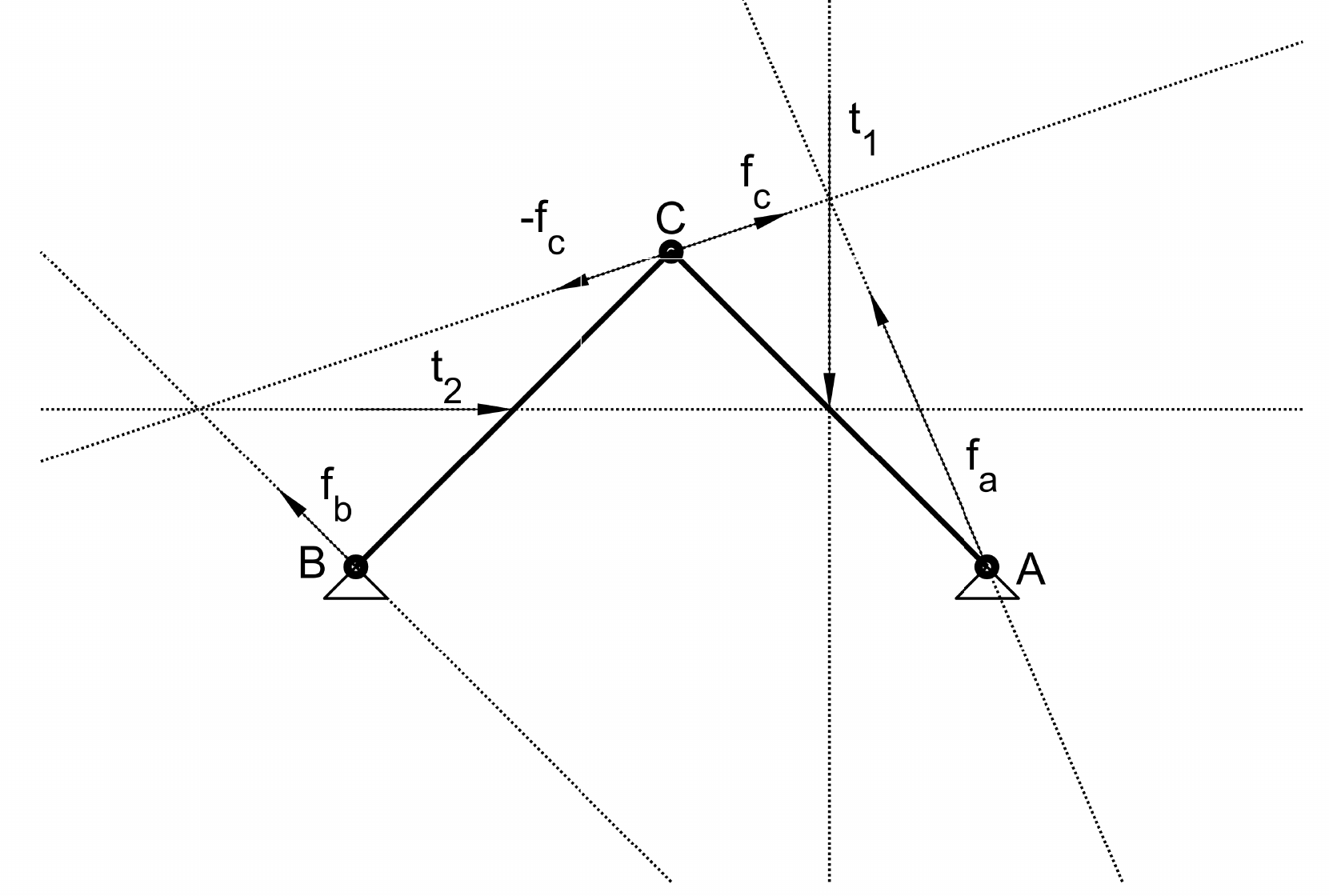}
\caption{Example three jointed structure, the dual of which is presented in Figure \ref{fig:i1}.} 
\label{fig:i2}
\end{figure}

\begin{figure*}[h!]
\centering
\includegraphics[width=0.8\textwidth]{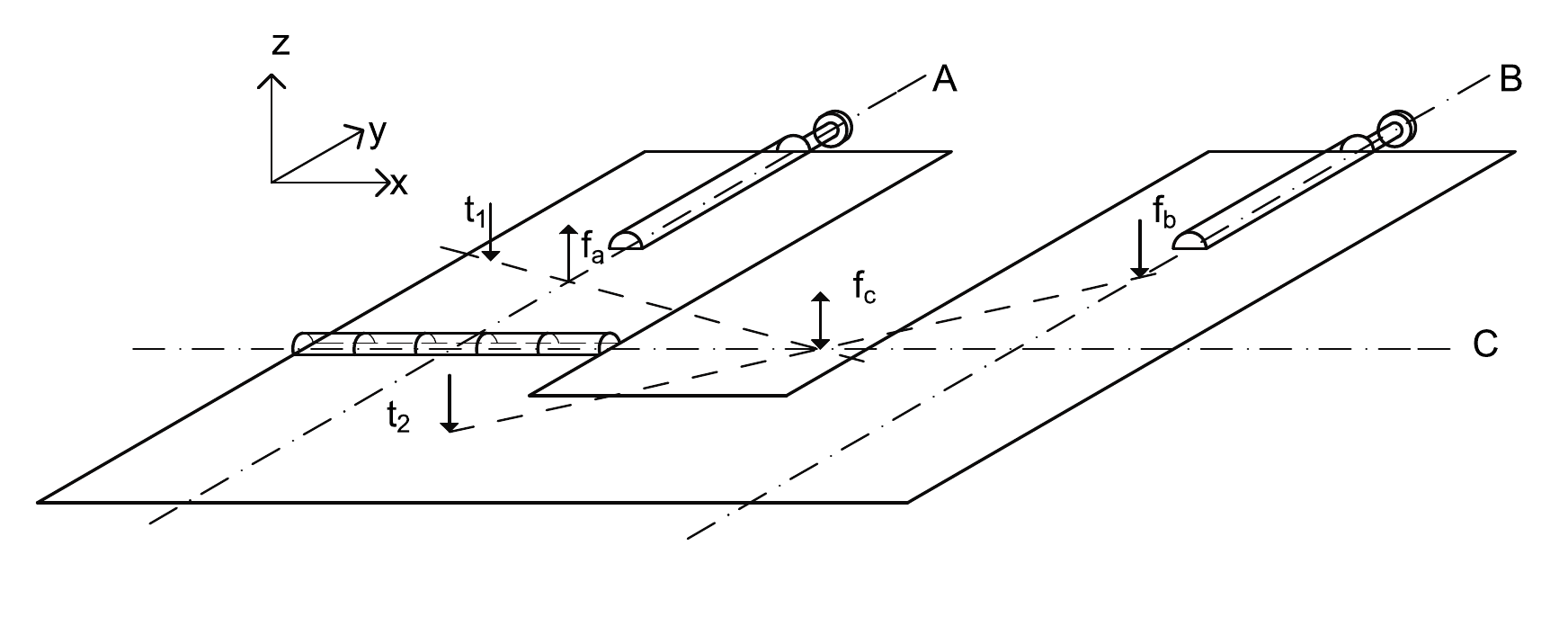}
\caption{Dual structure of the three jointed structure presented in Figure \ref{fig:i2}.}
\label{fig:i1}
\end{figure*}

\begin{corollary}
A planar structure's degrees of statical (and thus kinematical) in- or over-determinacy is invariant both under projective planar transformations and dualities.
\end{corollary}

\begin{proof}
Let us start by noting that applying a congruence does not change the geometry of the structure, meaning they have no influence on static and kinematic properties of structures. We only have to check the linear maps.\\
Degree of static over-determinacy: We can divide the indices of the force system of an arbitrary body into supporting forces $\mathcal{S}\subset\mathcal{I}$ (internal or external) and loads $\mathcal{L}:=\mathcal{I}\setminus\mathcal{S}$. Solving the equilibrium equations means determining values of $\phi_k$ such that 
\begin{align}
\sum_{\mathcal{L}}f_j+\sum_{\mathcal{S}}\phi_k\bar{f_k}=0
\end{align}  
holds, where $\bar{f_k}$ is a unit force in the appropriate position, according to the convention introduced in subsection \ref{sec:inner}. This can be done for all bodies of the structure. Any possible load $f_j$ leading to a contradiction in the arising system of equations is mapped to a possible load $f_jA$ leading to a contradiction in the image system of equations. As the maps are invertible, any pre-image can be considered an image under the inverse map, implying the other direction.

Degree of static indeterminacy: The images of the linearly dependent forces constituting self stresses are explicitly given by Theorems \ref{thm:projective} and \ref{thm:duality} as linearly dependent forces constituting self stresses. The invertibility of the transformation again gives the other direction.

As the kinematical properties of the structure are tied to the statical ones such that to each kinematical degree of inteterminacy there exists a degree of static overdeterminacy and vice versa \cite{roller_szabo}, the static properties imply the kinematic ones.
\end{proof}

Although this is a long standing result for bar and joint frameworks, plate (or panel/sheet) structures \cite{crapo1982statics,wunderlich1982projective} and grillages \cite{TARNAI1989duality}, for the general case of arbitrary bodies subjected to arbitrary forces this appears to be the first intentional proof. The one in \cite{crapo1982statics} presents many details of line geometry that would give the three dimensional general case had they decided to pursue it.

\subsubsection{Some graphic methods derived through duality}\label{sec:graphic}
One of the notable fields of mechanics where projective geometry is present is graphic statics, mostly developed for planar force systems. Here we will briefly see a few examples how the methods present in it can be used to describe other force and velocity systems. 

\begin{figure}[h!]
\centering
\includegraphics[width=0.8\columnwidth]{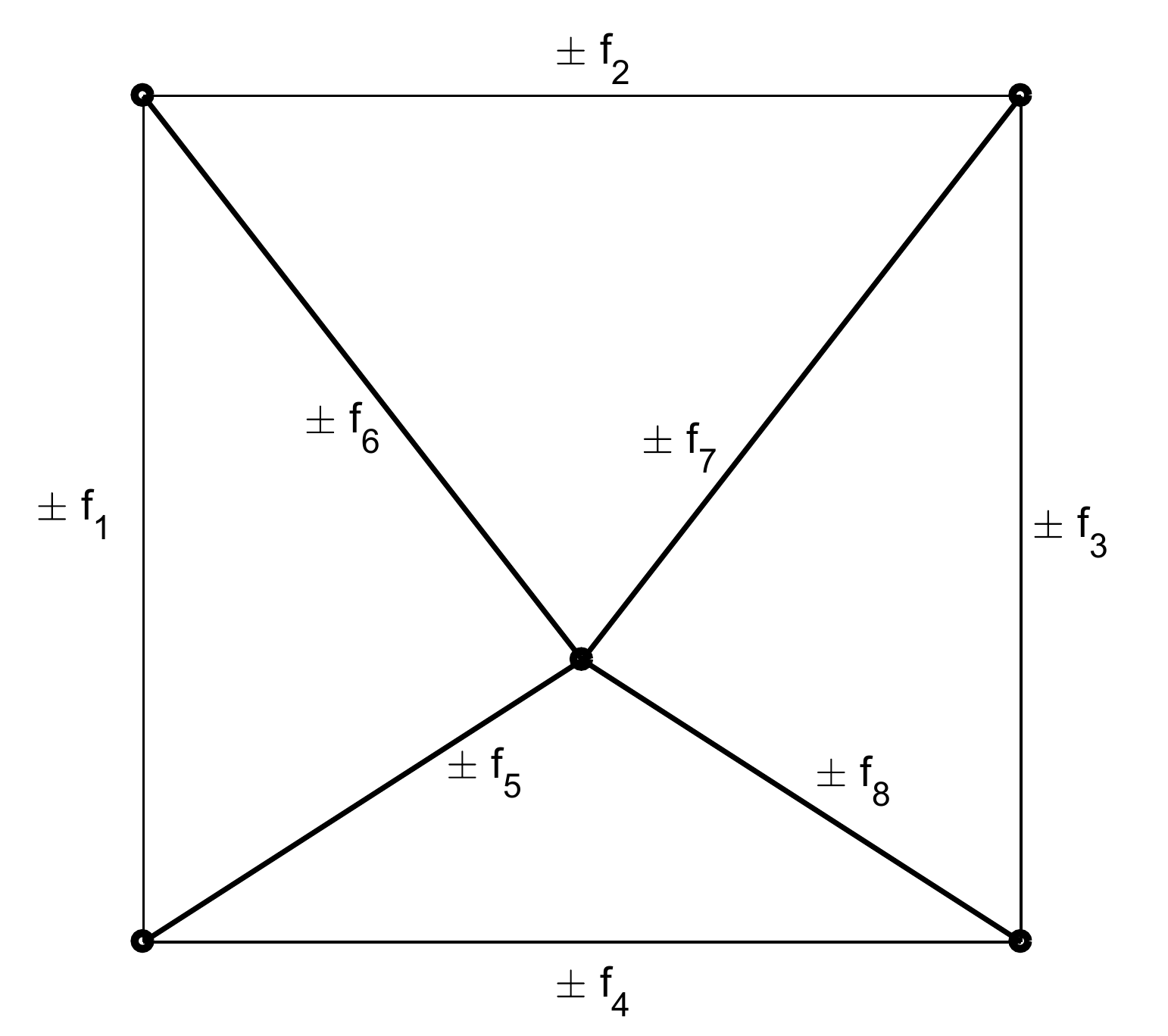}
\caption{Statically indeterminate truss, in a tensegrity set-up. The thin lines represent ropes under tension, the thick lines bars under compression.  The dual grillage is presented in Figure \ref{fig:g2}. Notation $\pm f_i$ is introduced, since each force effects its two endpoints differently: compare with the corresponding three dimensional force diagram in Figure \ref{fig:h}.}
\label{fig:g1}
\end{figure}

Grillages ($\mathcal{F}_c$) have been known\cite{TARNAI1989duality} to be the duals of trusses ($\mathcal{F}_p$) since at least the 1980-s. Through this duality it is known \cite{whiteley1991weavings}, that if a grillage possesses self stress there has to exist a closed a spherical polyhedron ($\mathcal{P}$) whose projection is the dual grid of the grillage. This dual grid in essence defines the geometry of a truss, whose bar forces correspond to the forces between the beams of the grillage, according to Theorem \ref{thm:duality}. As such, the existence of $\mathcal{P}$ is tied to Maxwell's theorem of self stresses in trusses. Faces of $\mathcal{P}$ are formed by planes that are the evaluations of linear combinations of the moment functionals of the truss, while there exists a dual polyhedron $\mathcal{Q}$ which can be considered as a three dimensional force plan of the truss: The vertices of $\mathcal{Q}$ are such that the projection of the edges (difference vectors) on the $x,y$ plane gives the Maxwell-Cremona force plan, while the $z$ coordinates give the moments of each force with respect to the origin. An analogous force plan can be created for the grillage. If the duality relating the two structures is given by the identity matrix, one merely has to relabel the axes of the diagram. An example of this can be seen in Figures \ref{fig:g1},\ref{fig:g2} and \ref{fig:h}.

\begin{figure}[h!]
\includegraphics[width=1\columnwidth]{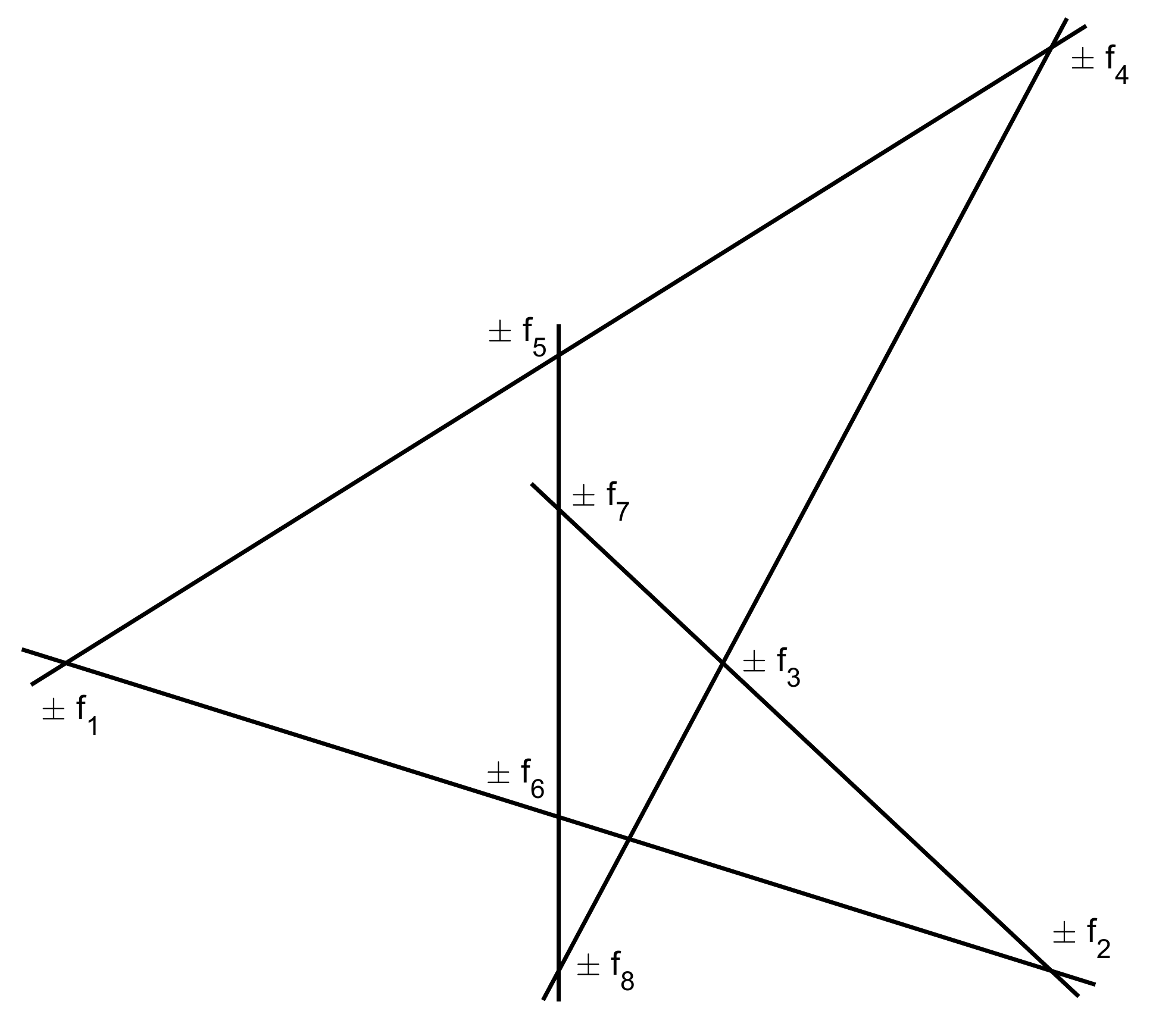}
\caption{Dual grillage of the tensegrity presented in Figure \ref{fig:g1}. At each intersection force $f_i$ effects one beam while force $-f_i$ effects the other one. The corresponding three dimensional force diagram is shown in Figure \ref{fig:h}.}
\label{fig:g2}
\end{figure}

Mechanisms of $\mathcal{E}_p$ possess the property, that the $z$ directional velocity at point $p$ caused by $d \in \mathcal{E}_p$ is $\scal{p}{d}$; meaning the velocity diagram of such mechanism can be conveniently graphed above the $x,y$ plane, resulting in a set of 3 dimensional planes (see planes $s_i$ in Figure \ref{fig:e}). In essence it is the dual of the moment functional of $\mathcal{F}_c$. We can repeat the known graphic statics methods and take a projective 3 dimensional dual of these planes, resulting in a 3 dimensional velocity diagram analogous to polyhedron $\mathcal{Q}$ introduced above. To each body corresponds a vertex representing its absolute velocity, while the relative velocities of the bodies can be measured as the difference vectors between these vertices (see Figure \ref{fig:f}).

Admittedly this type of mechanism is not the most widespread. Maybe gear trains with coplanar axes could be a subject it helps.

\begin{figure}[h!]
\includegraphics[width=1\columnwidth]{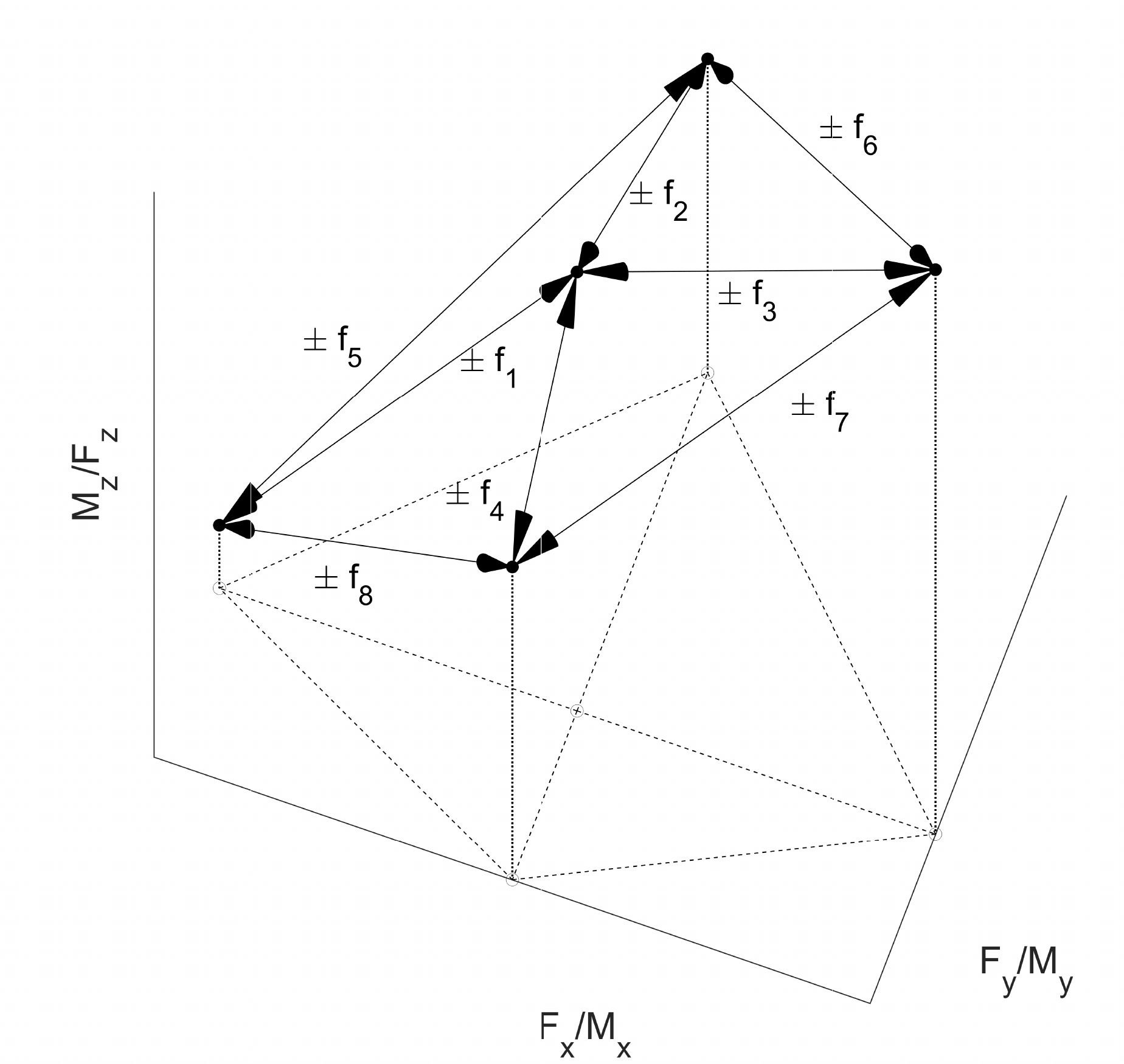}
\caption{Three dimensional force diagram corresponding to the truss in Figure \ref{fig:g1} and the grillage in Figure \ref{fig:g2}.}
\label{fig:h}
\end{figure}

\subsection{Static and kinematic chains}
The third types of transformations are interchanging velocities with forces, namely: $\mathcal{E}_p \leftrightarrow \mathcal{F}_p$, $\mathcal{E}_p \leftrightarrow \mathcal{F}_c$, $\mathcal{E}_c \leftrightarrow \mathcal{F}_c$ and $\mathcal{E}_c \leftrightarrow \mathcal{F}_p$. Here the correspondence between entire systems is not necessary straightforward, since for instance a body can be subjected to multiple forces while it can only possess a singe velocity state. There are several works \cite{shai2001duality,shai2002utilization,shai2006extension} born out of the graph theoretical description of certain static and kinematic systems, utilizing the duality concept coming from graph theory. They match several types of static systems (for example trusses, pillar systems and Stewart platforms) to corresponding types of mechanisms (for example planar linkages or serial robots). We will not attempt to rewrite those equations in our projective setting, but will restrict ourselves to a pair of systems: static and kinematic chains. We will however touch on the "graph theory based" dualization of the Aronhold-Kennedy theorem \cite{shai2006study} and propose a more basic dual-theorem.\\

Under a kinematic chain we will mean a set of $\{e_i\}$ ($i\in \{1\dots n\}$) absolute velocities corresponding to $n$ bodies linked serially. Let 
\begin{align}
e_{j,k}:=e_k-e_j\label{eq:kin_chain_def}
\end{align}
denote their relative velocities. In a real-life example typically the points/lines corresponding to the relative velocities are given, while determining magnitudes of these relative velocities and the absolute velocities is part of the question.\\

Under a static chain, we will mean a set of $\{f_i\}$ ($i\in \{1\dots n\}$) forces carried by some medium. At any point the force $f_i$ can be considered the resultant of some force system, and it changes according to the loads  $f_{j,k}$, as 
\begin{align}
f_{j,k}:=f_k-f_j. \label{eq:stat_chain_def}
\end{align}
Note how the equilibrium condition is embedded in this definition. One practical example of these is a literal chain, whose shape (lines of $f_i$) is dependant on the loads $(f_{j,k})$ it is subjected to, but it can also be the series of resultant forces corresponding to beams carrying a series of loads.

For these restricted systems, we can have the following statement:

\begin{prop}
\begin{enumerate}[label=(\roman*)]
\item[]
\item Invertible linear transformations  $\mathcal{E}_p \leftrightarrow \mathcal{F}_p$, and  $\mathcal{E}_c \leftrightarrow \mathcal{F}_c$ map compatible velocities of kinematic chains and forces of static chains that are in statical equilibrium into each other. There is a bijection between the equivalence classes of these transformations and projective transformations of $PG(2)$.
\item Invertible linear transformations  $\mathcal{E}_p \leftrightarrow \mathcal{F}_c$, and  $\mathcal{E}_c \leftrightarrow \mathcal{F}_p$ map compatible velocities of kinematic chains and forces of static chains that are in statical equilibrium into each other. There is a bijection between the equivalence classes of these transformations and projective dualities of $PG(2)$.
\end{enumerate}
\end{prop}

\begin{proof}
The properties of geometric transformation are essentially the same as before, the only thing to show is how compatibility and equilibrium imply each other. A kinematic chain is compatible if the velocities are such that equation \eqref{eq:kin_chain_def} holds for all $j,k\in \{1\dots n\}$. Transforming it with an invertible linear map, say $A$ gives
\begin{align}
e_{j,k}A=e_kA-e_jA=f_{j,k}=f_k-f_j
\end{align}
$\forall \ j,k\in \{1\dots n\},$ which is precisely \eqref{eq:stat_chain_def}, containing the equilibrium conditions. Matrix $A$ is invertible by definition, meaning the other direction is similar.
\end{proof}

\begin{figure}[h]
\includegraphics[width=1\columnwidth]{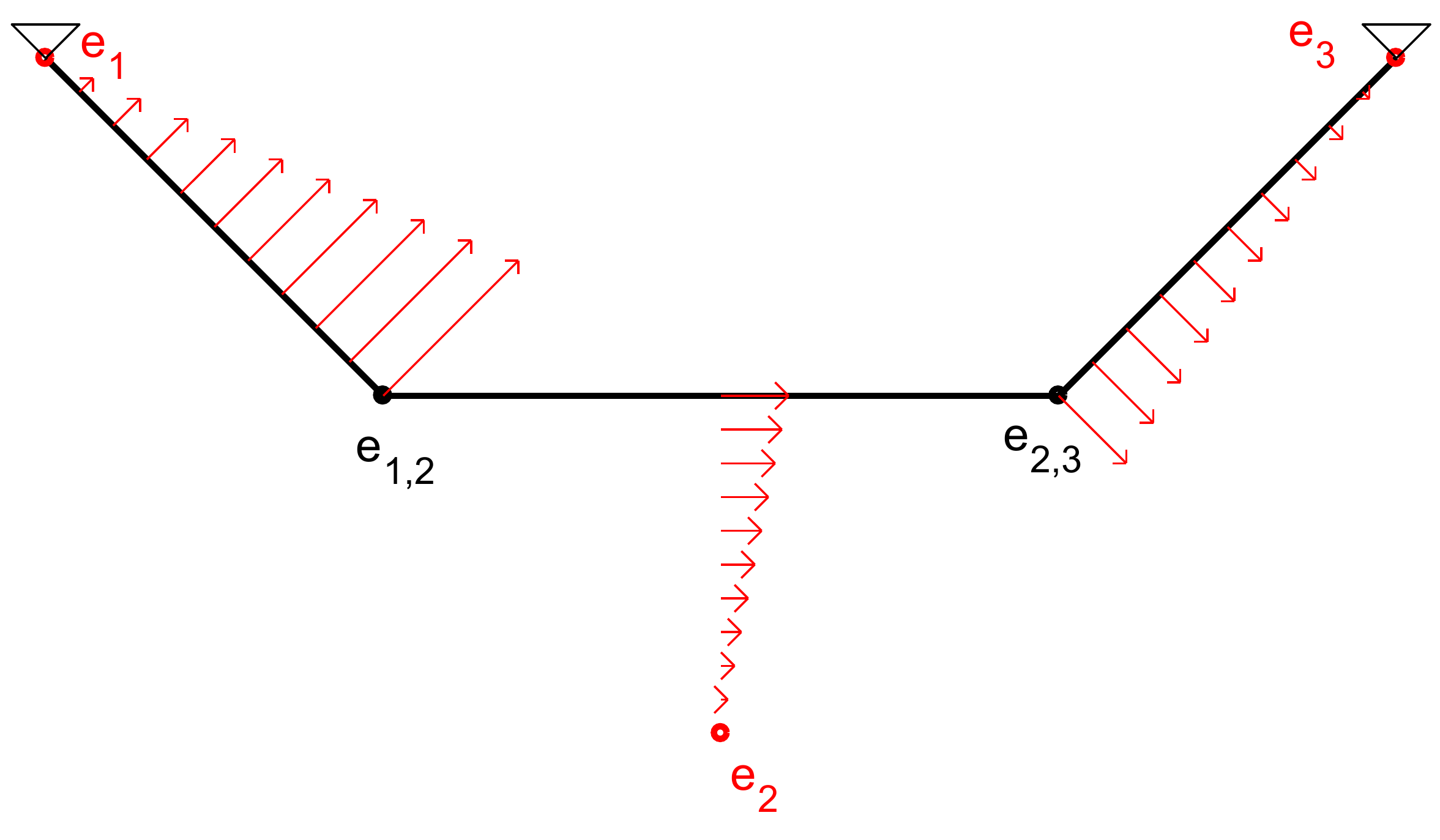}
\caption{Four bar linkage mechanism (also, a kinematic chain), $e_{1,3}$ lies at infinity, where lines $\overline{e_1e_3}$ and $\overline{e_{1,2}e_{2,3}}$ meet. The arrows indicate the velocities of the appropriate points of the bodies.}
\label{fig:d}
\end{figure}

\begin{figure*}[h]
\centering
\includegraphics[width=0.8\textwidth]{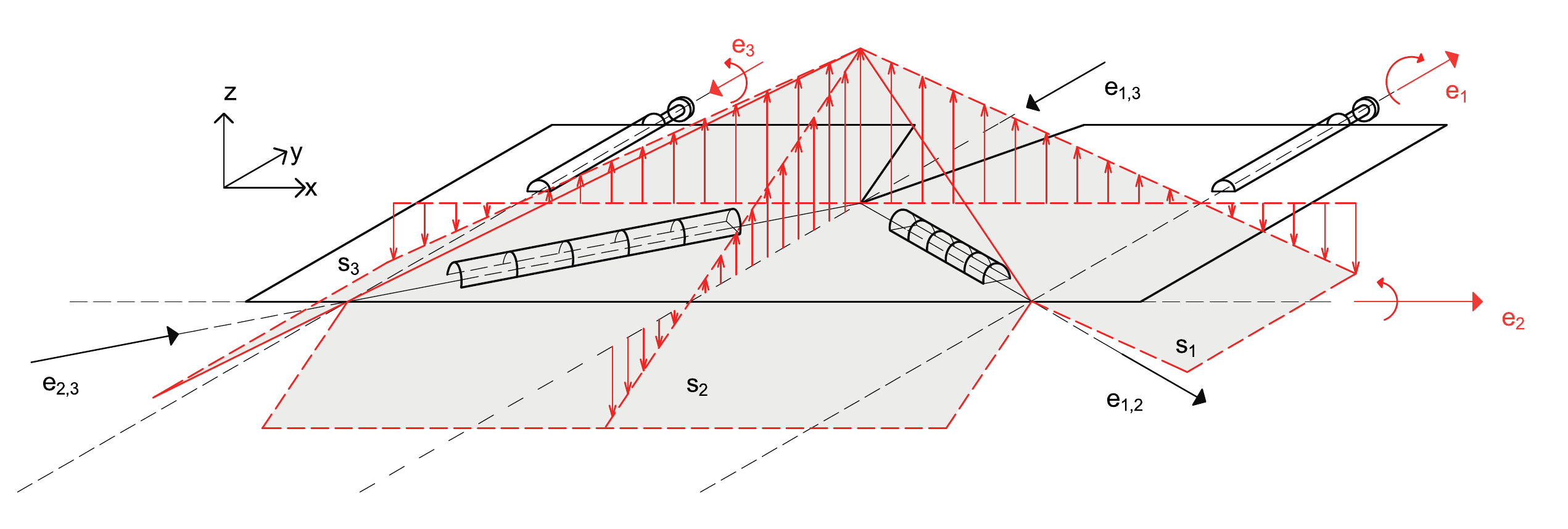}
\caption{Dual mechanism of the kinematic chain shown in Figure \ref{fig:d}. The plates are supported by linear hinges. The arrows orthogonal to the plane indicate the velocities of the appropriate points of the bodies. Planes $s_i$ are the evaluations of  the $z$ directional velocity functions $v_z(p)=\scal{e_i}{p}$.}
\label{fig:e}
\end{figure*}

\begin{figure}[h]
\includegraphics[width=1\columnwidth]{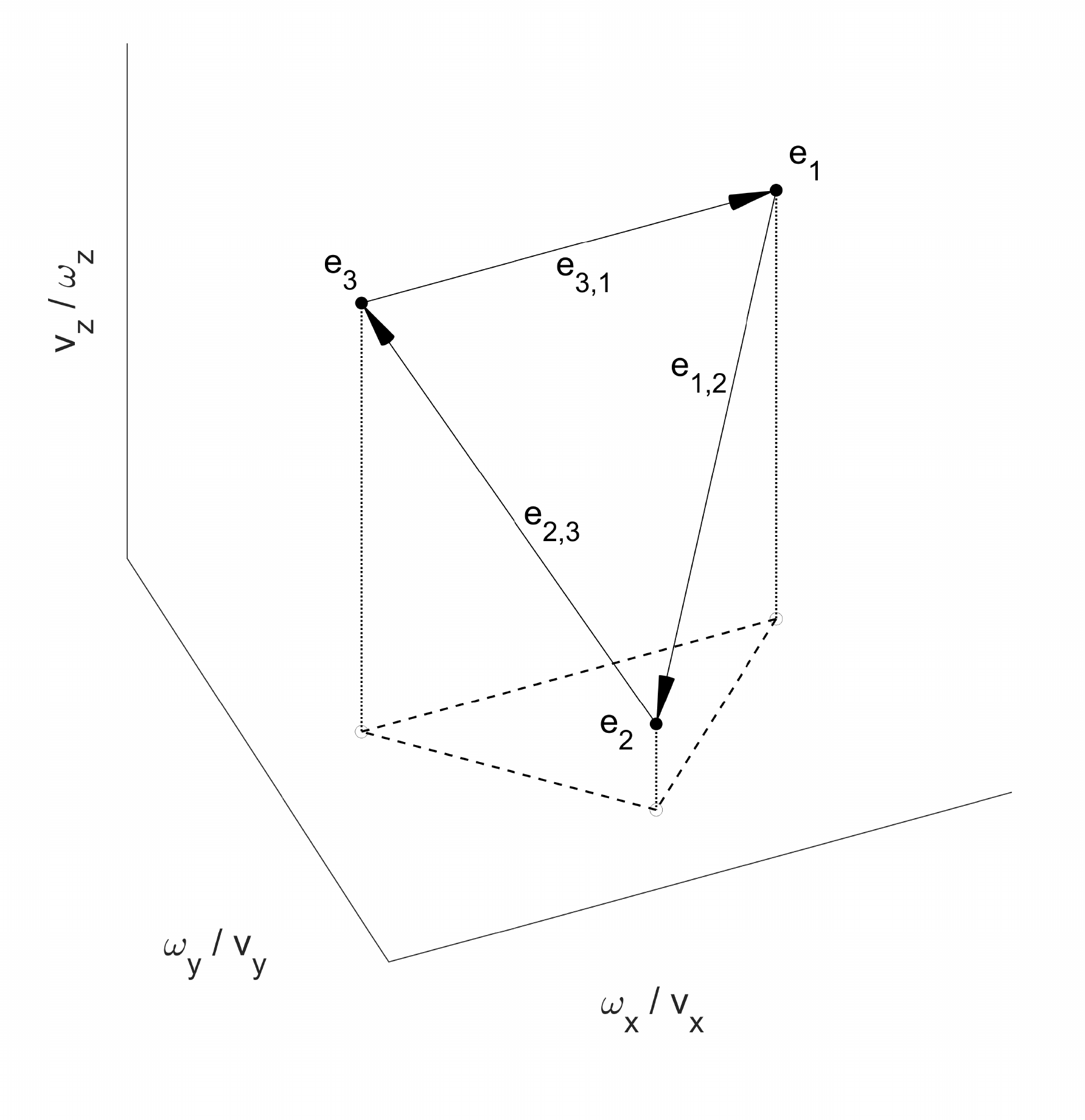}
\caption{Diagram of the velocity states of the mechanisms shown in Figure \ref{fig:d} and in \ref{fig:e}.}
\label{fig:f}
\end{figure}

\subsubsection{The Aronhold-Kennedy theorem and its duals}
One of the often cited theorems in kinematics is the Aronhold-Kennedy theorem \cite{kinematicsbook}.
\begin{theorem}[Aronhold-Kennedy]
If three bodies move in the plane relative to each other, their relative instant centres (the points around which they are rotating relative to each other) are collinear.
\end{theorem}

A dual theorem trivially appears:
\begin{corollary}\label{thm:corr_aron_dual}
If three bodies move with coplanar axes of rotations, their relative axes of rotation are concurrent.
\end{corollary}
\begin{proof}
Follows from the principle of duality of $PG(2)$.
\end{proof}
An illustration of this can be seen in Figure \ref{fig:e}.
There is a dualization to forces present in literature, where the authors introduced the concept of the equimomental line:

\begin{definition}[Equimomental line]
For two arbitrary forces acting in a single plane, there exists a unique line -the equimomental line- in the plane where the moments about each point on the line, due to the two forces, are equal.
\end{definition}

\begin{theorem}[Shai-Pennock \cite{shai2006study}]
The three equimomental lines defined by three arbitrary coplanar forces must intersect at a unique point.  
\end{theorem}

Here a more fundamental formulation is presented, using only the basic concept of the force:
\begin{theorem}\label{thm:diff}
Given three coplanar forces, their pairwise differences are concurrent.
\end{theorem} 
\begin{proof}
Denoting the forces with $f_1$, $f_2$, $f_3$ $\in \mathcal{F}_p$, their pairwise differences are
\begin{align}
f_{1,2}&=f_2-f_1 \label{eq:af1} \\ 
f_{2,3}&=f_3-f_2 \ \label{eq:af2}\\
f_{3,1}&=f_1-f_3.
\end{align}
Adding \eqref{eq:af1} to \eqref{eq:af2} gives

\begin{align}
f_{1,2}+f_{2,3}=f_3-f_1&=-f_{3,1} \\
f_{1,2}+f_{2,3}+f_{3,1}&=0
\end{align}
which is precisely the concurrency condition of three lines in $PG(2)$.
\end{proof}

From the proof it is apparent that this is in fact an equilibrium condition: the fact that that three coplanar forces are in equilibrium if and only if their lines of action are concurrent is part of engineering education since at least Bow \cite{bow2014economics} and Culmann \cite{culmann1875graphische}. 
\begin{remark}
It is easy to see how the equimomental line corresponding to two forces coincides with the line of action of the difference of the two forces.
\end{remark}

In our context, a static-dual theorem of this again appears trivially:
\begin{corollary}
Given three forces orthogonal to a plane, the points of attack of their pairwise differences are collinear on the plane.
\end{corollary}
\begin{proof}
Follows from the principle of duality of $PG(2)$.
\end{proof}  
However, this can be formulated in a stronger way, whose projective special case is the corollary given above:
\begin{theorem}
Given three concurrent forces in $3D$, their pairwise differences are coplanar.
\end{theorem}
\begin{proof}
Identical to the proof of Theorem \ref{thm:diff}, with the distinction that the vectors involved are actual force vectors in $3D$.
\end{proof}  

\subsubsection{Static chains and moment diagrams}
One of the typical examples of static chains are the series of resultant forces inside a beam or frame. As engineers are usually interested in the moment diagram of the structure, a graphic way of obtaining it for certain loads is presented here. Based on oral discussions this phenomenon is not unknown to thrust-line researchers, but the literature seems to contain it in less general forms for arches \cite{fuller1675arch}, or straight beams \cite{wolfe1921graphical}. We unify the two approaches, by keeping the unconstrained shape of the arches and allowing supports in the middle of the structure, similarly to beams. The applicability requires the assumption that if the structure is subjected to force (load or support force) $f_{j,k}$ at point $p$, then $p$ must lie on the line of $f_{j,k}$. In engineering practice this is a reasonable assumption.

\begin{prop}
Given a structure with the shape of a curve $p(t)$ $(t\in [0,1])$, subjected to forces such that in the arising static chain all $f_{j,k}$ forces are parallel and $f_{j,k}$ is not parallel to $f_{i}$, then 
\begin{enumerate}[label=(\roman*)]
\item it is possible to pair each point $p(t)$ of the structure with point $q(t)$ along the line of the static chain, such that $q(t)$ lies on the line of action of the force the structure is subjected to at point $p(t)$.  

\item line segments $\overline{p(t),q(t)}$ are parallel to $f_{j,k}$, and their signed length is proportional to the value of the moment function describing the bending moment the structure is subjected to at each point $p(t)$.
\end{enumerate}
\end{prop}

\begin{proof}
\begin{enumerate}[label=(\roman*)]
\item[]
\item As per the starting assumption, if the structure is subjected to force (load or support force) $f_{i-1,i}$ at point $p(t_1)$, then  $p(t_1)$ must lie on the line of $f_{i-1,i}$. By the definition of the pairing, point $q(t_1)$ has to lie at the intersection point of $f_i$ and $f_{i-1}$. The statical equilibrium $f_{i}=f_{i-1}+f_{i-1,i}$ implies $q(t_1)$ also lies on the line $f_{i}$. A similar argument can be stated about points $p(t_2)$, $q(t_2)$ and force $f_{i,i+1}$, which is parallel to $f_{i-1,i}$ according to one of the validity criteria of the theorem. This implies a parallel projection between each point of $p(t)$ and $q(t)$ between $t_1$ and $t_2$, and it can be repeated on all segments $(t \in [0,1])$. This projection is the desired pairing.
\item Let us represent the ideal point in the direction of $f_{i,j}$ with 
\begin{align}
u:=(u_x,u_y,0) \text{ such that } \norm{u}=1 \label{eq:ideal_u}
\end{align}
holds. For any two points $p(t_1)$ and $p(t_2)$ of the structure sharing the property that they are subjected to the same resultant $f_i$, the appropriate values of the moment function are 
\begin{align}
m(t_1)=\scal{p(t_1)}{f_i}\label{eq:moments1}
\end{align}
and
\begin{align}
m(t_2)=\scal{p(t_2)}{f_i}. \label{eq:moments2}
\end{align}
Due to the collinearity of the three points we can have
\begin{align}
q(t_1)=p(t_1)+\lambda(t_1)u. \label{eq:linkomb1}
\end{align}
After inspecting \eqref{eq:ideal_u} it is apparent that the signed distance of $p(t_1)$ and $q(t_1)$ is precisely $\lambda(t_1)$. Its value can be calculated using $\scal{q(t_1)}{f_i}=0$ as
\begin{align}
\lambda(t_1)=-\frac{\scal{p(t_1)}{f_i}}{\scal{u}{f_i}}. \label{eq:lambda1}
\end{align}
This will never give division by zero, due to one of the conditions of the theorem: $f_i \nparallel f_{j,k} \implies \scal{u}{f_i} \neq 0$ (non-parallel forces meet at a finite point). After repeating equations \eqref{eq:linkomb1} and \eqref{eq:lambda1} using parameter value $t_2$, and substituting everything into equations \eqref{eq:moments1} and \eqref{eq:moments2}, they can be rearranged into:
\begin{align}
\frac{m(t_1)}{m(t_2)}=\frac{\lambda(t_1)}{\lambda(t_2)}
\end{align}
which is the desired statement.
\end{enumerate}
\end{proof}

\begin{remark}
This works even if the structure is parallel to the direction of $f_{j,k}$, although in this case the line segments overlap. 
\end{remark}

\begin{remark}
It is easy to see how this remains valid in case of distributed forces, where $q(t)$ turns into a $C_1$ continuous curve.
\end{remark}

\begin{remark}
In case of cantilevers the end load need not be unidirectional, as it plays the role of ${f_n}$. On the other hand supports in the middle of the structure play the role of $f_{j,k}$ and have to be unidirectional.
\end{remark}

\begin{figure}[h!]
\centering
\includegraphics[width=0.75\columnwidth]{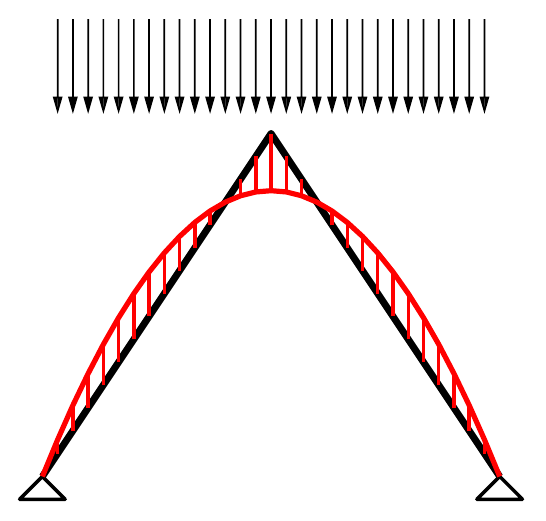}
\caption{In case of unidirectional loads the moment diagram appears between the line of action of the resultant and the structure, if the two curves are paired in the appropriate direction.}
\label{fig:c}
\end{figure}

The natural question arises: what happens when we lift the directional constraint on $f_{j,k}$? If we want to coherently pair each point $p(t)$ to a single point $q(t)$ on the force-chain (in other words we want a continuous $q(t)$ curve), we have to rely on the intersection points ($o_i$) of $f_{i-1,i}$ and $f_{i,i+1}$ to determine the points from which  points $p(t)$ can be projected centrally to $f_i$, giving points $q(t)$.
Such construction can be seen in Figure \ref{fig:a}. The gain with lifting the restriction can be summed up in the following proposition (illustrated in Figure \ref{fig:b}):

\begin{figure}[h!]
\includegraphics[width=1\columnwidth]{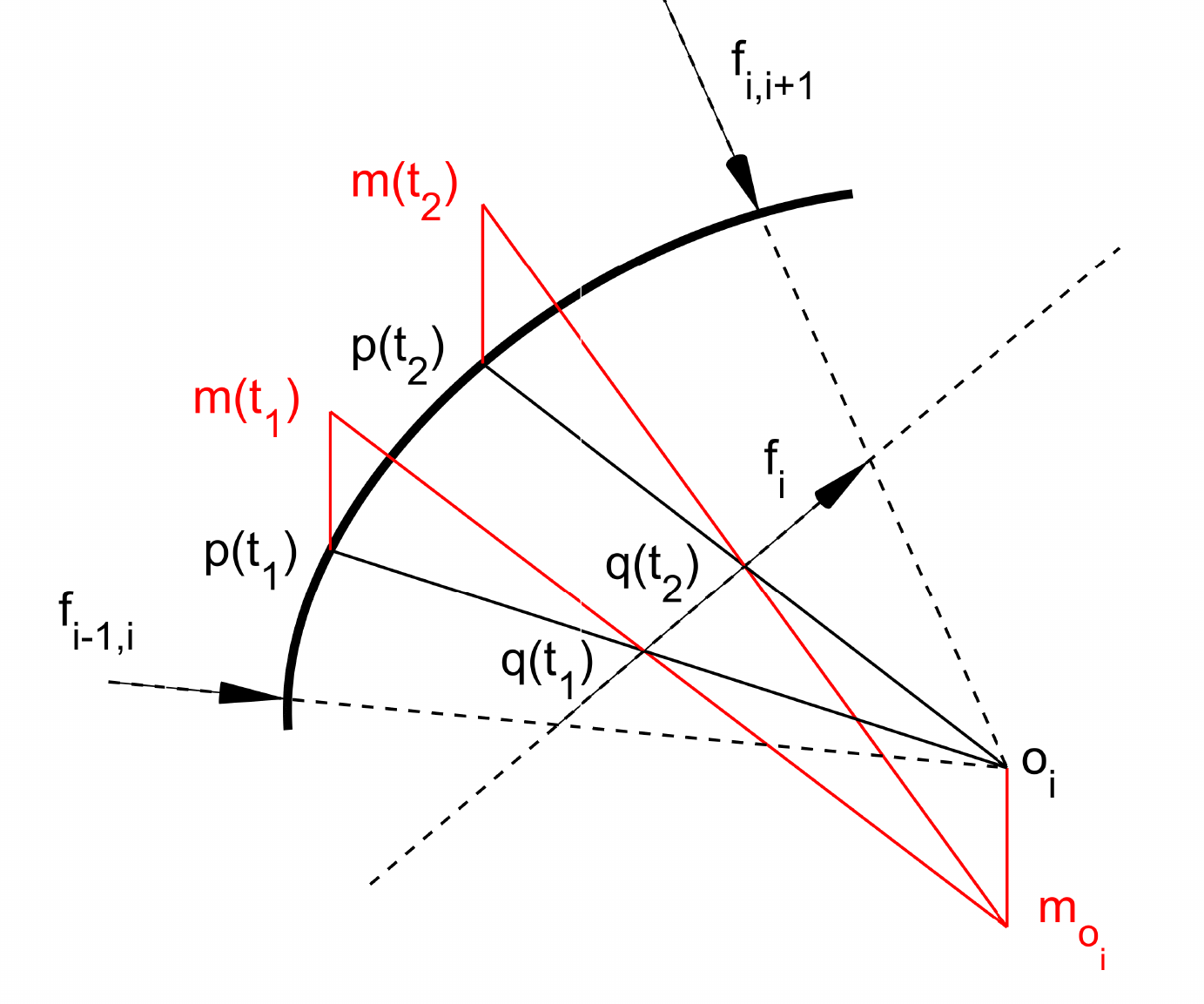}
\caption{Moment diagram candidate, in case of non-parallel loads. The structure is shown in thick black, while the moment values in red. (Axonometric figure, the moment values are measured in the vertical direction.)}
\label{fig:a}
\end{figure}

\begin{prop}
Given a structure with the shape of a curve $p(t)$ $(t\in [0,1])$, supported only on its ends and subjected to a radial load uniformly distributed along the circumference of a circle with the corresponding radii, then 
\begin{enumerate}[label=(\roman*)]
\item it is possible to pair each point $p(t)$ of the structure with point $q(t)$ along the curve of the static chain, such that $q(t)$ lies on the line of action of the force the structure is subjected to at point $p(t)$.  
\item line segments $\overline{p(t),q(t)}$ are parallel to the radial direction of the load, and their signed length is proportional to the value of the moment function describing the bending moment the structure is subjected to at each point $p(t)$.
\end{enumerate}
\end{prop}

\begin{proof}
Denoting the moment $\scal{f_i}{o_i}$ with $m_{o_i}$ and the signed distances of points $p$ and $q$ with $\overrightarrow{p,q}$, we can observe that

\begin{align}
\frac{\overrightarrow{q(t_1),p(t_1)}}{\overrightarrow{q(t_1),o_i}}=\frac{m(t_1)}{m_{o_i}}
\end{align}
and 
\begin{align}
\frac{\overrightarrow{q(t_2),p(t_2)}}{\overrightarrow{q(t_2),o_i}}=\frac{m(t_2)}{m_{o_i}}.
\end{align}

This can be rearranged into 
\begin{align}
\frac{m(t_1)}{m(t_2)}=\frac{\overrightarrow{q(t_1),p(t_1)}}{\overrightarrow{q(t_2),p(t_2)}}\frac{\overrightarrow{q(t_2),o_i}}{\overrightarrow{q(t_1),o_i}},
\end{align}
meaning this construction gives a moment diagram whenever 
\begin{align}
\frac{\overrightarrow{q(t_2),o_i}}{\overrightarrow{q(t_1),o_i}}=1
\end{align}
holds. The case of parallel $f_{j,k}$ can be considered such special case, where the two "infinite distances" cancel each other out. The other case when this holds is the case when the line of action is a circle, which cannot hold in case of concentrated forces but is possible with uniformly distributed radial load.
\end{proof}

\begin{figure}[h!]
\includegraphics[width=1\columnwidth]{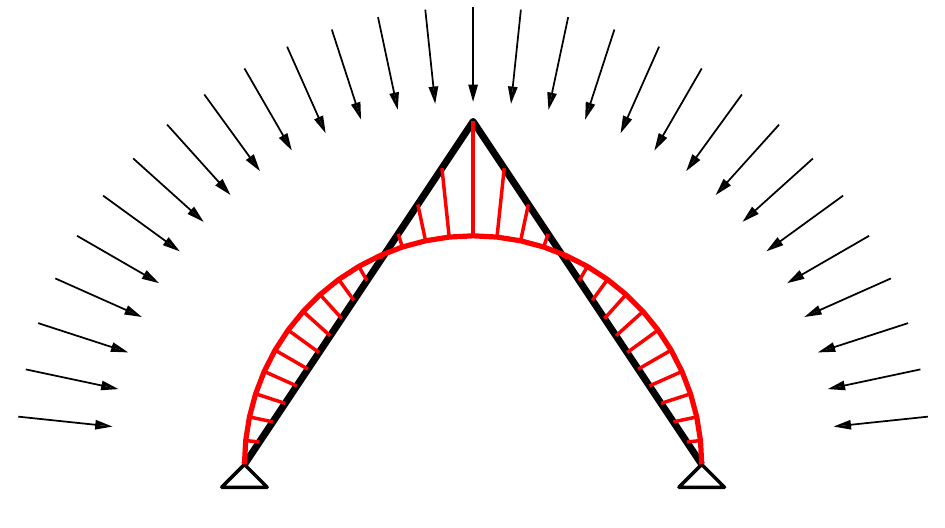}
\caption{In case of uniformly distributed radial loads the lines of action of the resultants form a circle, and the moment diagram appears between the circle and the structure.}
\label{fig:b}
\end{figure}

\section{Summary}
A clear cut description of transformations of planar force, velocity and displacement systems was presented. It was shown that for structures where the geometry uniquely (up to scaling) determines the behaviour, there is a one-to one correspondence between projective geometric transformations and dualities of planar mechanical problems, and scaling induced equivalence classes of linear transformations of the force/velocity/displacement systems. For structures outside this class, another projective tool is needed: congruences of force / velocity / displacement systems. The usefulness of this description was demonstrated in a few examples:

The projective invariance of rigidity was proven for all types of planar structures (extended from bar-joint frameworks and grillages).

A past combinatorial result regarding self stresses of grillages  and spherical polyhedra were analytically supplemented.

An exact criterion was given  to the moment diagram of a structure appearing between the reference curve of the structure and the lines of action of the resultant forces (thrust-line). This gives a possibility to graph the moment diagrams in an intuitive yet precise way.

The static dual of the Aronhold-Kennedy theorem of kinematics was given without the need to define concepts additional to the concept of the force.\\

The author hopes that the ease of this description helps educators when teaching the connection between mechanics and geometry as well as researchers when optimizing structures through graphical transformations.

\section{Acknowledgements}
The author wishes to thank O. G\'asp\'ar and Dr. P.L. V\'arkonyi for the useful discussions on thrust-lines and the presentation of this work.

\bibliographystyle{unsrt} 
\bibliography{grafobib}

\end{document}